\documentclass[12pt, conference, compsocconf]{article}
\usepackage{amsmath}
\usepackage{amsfonts}
\usepackage{amssymb}
\usepackage{dsfont} 
\usepackage[standard]{ntheorem}
\usepackage{epsfig}
\usepackage{fullpage}

\newcommand{\llbracket}{[\![}
\newcommand{\rrbracket}{]\!]}

\begin{document}
%
\title{\Huge{Chaotic iterations versus Spread-spectrum: topological-security and
stego-security}}


\author{Christophe Guyeux, Nicolas Friot, and Jacques M. Bahi\\
\\
Computer Science Laboratory LIFC\\
University of Franche-Comt\'{e}\\
rue Engel Gros, Belfort, France\\
\\
\{christophe.guyeux, nicolas.friot, jacques.bahi\}@lifc.univ-fcomte.fr\\
}


%



\maketitle

Spread-spectrum.

\begin{abstract}
A new framework for information hiding security, called topological-security,
has been proposed in a previous study. It is based on the evaluation of unpredictability
of the scheme, whereas existing notions of security, as stego-security, are more
linked to information leaks. It has been proven that spread-spectrum techniques,
a well-known stego-secure scheme, are topologically-secure too. In this paper,
the links between the two notions of security is deepened and the usability of
topological-security is clarified, by presenting a novel data hiding scheme that
is twice stego and topological-secure. This last scheme has better scores than
spread-spectrum when evaluating qualitative and quantitative topological-security
properties. Incidentally, this result shows that the new framework for security
tends to improve the ability to compare data hiding scheme.

\end{abstract}

\newpage

%

\section{Introduction}\label{sec:introduction}

Information hiding has recently become a major digital technology~\cite{1411349},~\cite{Xie09:BASecurity}, especially with the increasing importance and widespread distribution of digital media through the Internet.
Spread-spectrum data-hiding techniques have been widely studied in recent years
under the scope of security. These techniques encompass several schemes, such as 
Improved Spread Spectrum~(ISS), Circular Watermarking~(CW), and Natural 
Watermarking~(NW). Some of these schemes have revealed various 
security issues. On the contrary, it has
been proven in~\cite{Cayre2008} that the Natural Watermarking technique is 
stego-secure. This stego-security is one of the security classes defined 
in~\cite{Cayre2008}. In this paper, probabilistic models are used to 
categorize the security of data hiding algorithms in the Watermark Only 
Attack~(WOA) framework.

We will show that the security level of such algorithms can be studied into a 
novel framework based on unpredictability, as it is understood in the 
theory of chaos~\cite{Devaney}. To do so, a new class of security will be 
introduced, namely the topological-security. This new class can be used to
study some categories of attacks that are difficult to investigate in the existing 
security approach. It also enriches the variety of qualitative and quantitative 
tools that evaluate how strong the security is, thus reinforcing the confidence 
that can be had in a given scheme. 

In addition of being stego-secure, it has been proven in~\cite{ih10} that 
Natural Watermarking technique is topologically-secure. Moreover, this technique 
possesses additional properties of unpredictability, namely, strong transitivity,
topological mixing, and a constant of sensitivity equal to $\frac{N}{2}$.
However NW are not expansive, which is problematic in the Constant-Message 
Attack (CMA) and Known Message Attack (KMA) setups~\cite{ih10}. In this paper, it is proven  
by using the new topological-security framework, that a more secure scheme than
NW can be found to withstand attacks in these setups. This scheme, introduced in~\cite{guyeux10ter}, is 
based on the so-called chaotic iterations. The aim of this work is to 
prove that this algorithm is stego-secure and topologically-secure, to study its
qualitative and quantitative properties of unpredictability, and then to compare
it with Natural Watermarking.


The rest of this paper is organized as follows. In Section~\ref{sec:basic-recalls}, 
basic definitions and terminologies in the field of topology, chaos, and
security are recalled. In Section~\ref{sec:chaotic iterations-security-level} the stego-security of chaotic
iterations is established in some cases, whereas in
Section~\ref{sec:chaos-security-evaluation} is studied the topological-security
of chaotic iterations. Natural Watermarking and chaotic iterations are then compared in Section~\ref{sec:comparison-application-context}. The paper ends with a conclusion where our contribution is summarized, and 
planned future work is discussed.

\section{Basic recalls}\label{sec:basic-recalls}

\subsection{Chaotic iterations}
In this section, the definition and main properties of chaotic iterations are recalled~\cite{bg10:ij}.

\subsubsection{Chaotic iterations}
\label{sec:chaotic iterations}

In the sequel $S^{n}$ denotes the $n^{th}$ term of a sequence $S$ and $V_{i}$
the $i^{th}$ component of a vector $V$. Finally, the following notation
is used: $\llbracket1;N\rrbracket=\{1,2,\hdots,N\}$.

Let us consider a \emph{system} of a finite number $\mathsf{N}$ of elements (or
\emph{cells}), so that each cell has a boolean \emph{state}. A sequence of length
$\mathsf{N}$ of boolean states of the cells corresponds to a particular
\emph{state of the system}. A sequence which elements belong to $\llbracket
1;\mathsf{N} \rrbracket $ is called a \emph{strategy}. The set of all strategies
is denoted by $\mathbb{S}.$

\begin{definition}
\label{Def:chaotic iterations}

The set $\mathds{B}$ denoting $\{0,1\}$, let
$f:\mathds{B}^{\mathsf{N}}\longrightarrow \mathds{B}^{\mathsf{N}}$ be a function
and $S\in \mathbb{S}$ be a strategy. The so-called \emph{chaotic iterations} are
defined by $x^0\in \mathds{B}^{\mathsf{N}}$ and $\forall (n,i) \in
\mathds{N}^{\ast} \times \llbracket1;\mathsf{N}\rrbracket$:
\begin{equation*}
x_i^n=\left\{
\begin{array}{ll}
x_i^{n-1} & \text{ if }S^n\neq i \\
\left(f(x^{n-1})\right)_{S^n} & \text{ if }S^n=i.\end{array}\right.\end{equation*}
\end{definition}


\subsubsection{Devaney's chaotic dynamical systems}
\label{subsection:Devaney}

Consider a metric space $(\mathcal{X},d)$ and a continuous function $f$ on
$\mathcal{X}$. $f$ is said to be \emph{topologically transitive} if, for any pair
of open sets $U,V\subset \mathcal{X}$, there exists $k>0$ such that $f^{k}(U)\cap
V\neq\varnothing $. $(\mathcal{X},f)$ is said to be \emph{regular} if the set of
periodic points is dense in $\mathcal{X}$. $f$ has \emph{sensitive dependence on
initial conditions} if there exists $\delta >0$ such that, for any $x\in
\mathcal{X}$ and any neighborhood $V$ of $x$, there exists $y\in V$ and
$n\geqslant 0$ such that $|f^{n}(x)-f^{n}(y)|>\delta $. $\delta $ is called the
\emph{constant of sensitivity} of $f$. Quoting Devaney in~\cite{Devaney},

\begin{Definition}
A function $f:\mathcal{X}\longrightarrow \mathcal{X}$ is said to be \emph{chaotic} on $\mathcal{X}$ if $(\mathcal{X},f)$ is regular, topologically transitive and has sensitive dependence on initial conditions.
\end{Definition}


\subsubsection{Chaotic iterations and Devaney's chaos}
\label{sec:topological}

In this section we give outline proofs of the properties on which our secure data
hiding scheme is based. The complete theoretical framework is detailed
in~\cite{bg10:ij}.

Denote by $\Delta $ the \emph{discrete boolean metric},
$\Delta(x,y)=0\Leftrightarrow x=y.$ Given a function $f$, define the
function: $F_{f}: \llbracket1;\mathsf{N}\rrbracket\times
\mathds{B}^{\mathsf{N}} \longrightarrow \mathds{B}^{\mathsf{N}}
$ such that $F_{f}(k,E)=\left( E_{j}.\Delta (k,j)+f(E)_{k}.\overline{\Delta
(k,j)}\right)_{j\in \llbracket1;\mathsf{N}\rrbracket}.
$

Let us consider the phase space
$\mathcal{X}=\llbracket1;\mathsf{N}\rrbracket^{\mathds{N}}\times
\mathds{B}^{\mathsf{N}}$ and the map $G_{f}\left( S,E\right) =\left( \sigma (S),F_{f}(i(S),E)\right)
$, where $\sigma$ is defined by $\sigma :(S^{n})_{n\in \mathds{N}}\in \mathbb{S}\rightarrow (S^{n+1})_{n\in \mathds{N}}\in \mathbb{S}$,
and $i$ is the map $i:(S^{n})_{n\in \mathds{N}}\in \mathbb{S}\rightarrow S^{0}\in
\llbracket1;\mathsf{N}\rrbracket$. So the chaotic iterations can be described by the following iterations:
$$X^{0}\in \mathcal{X}\text{ and }X^{k+1}=G_{f}(X^{k}).$$

We have defined in~\cite{bg10:ij} a new distance $d$ between two points $(S,E),(\check{S},\check{E} )\in \mathcal{X}$
by
$d((S,E);(\check{S},\check{E}))=d_{e}(E,\check{E})+d_{s}(S,\check{S}),$
where:
\begin{itemize}
\item
$\displaystyle{d_{e}(E,\check{E})}=\displaystyle{\sum_{k=1}^{\mathsf{N}}\Delta
(E_{k},\check{E}_{k})} \in \llbracket 0 ; \mathsf{N} \rrbracket$
\item
$\displaystyle{d_{s}(S,\check{S})}=\displaystyle{\dfrac{9}{\mathsf{N}}\sum_{k=1}^{\infty
}\dfrac{|S^{k}-\check{S}^{k}|}{10^{k}}} \in [0 ; 1].$
\end{itemize}

It is then proven that,



\begin{proposition}
\label{Prop:continuite} $G_f$ is a continuous function on $(\mathcal{X},d)$.
\end{proposition}

In the metric space $(\mathcal{X},d)$, the vectorial negation $f_{0} :\ 
\mathbb{B}^N  \longrightarrow  \mathbb{B}^N $, $(b_1,\cdots,b_\mathsf{N}) 
\longmapsto (\overline{b_1},\cdots,\overline{b_\mathsf{N}})$ satisfies the three
conditions for Devaney's chaos: regularity, transitivity, and
sensitivity~\cite{bg10:ij}. So,

\begin{proposition}
$G_{f_0}$ is a chaotic map on $(\mathcal{X},d)$ according to Devaney.
\end{proposition}

\subsection{Using chaotic iterations as information hiding schemes}\label{sec:data-hiding-algo-chaotic iterations}

\subsubsection{Presentation of the scheme}

We have proposed in~\cite{guyeux10ter} to use chaotic iterations as an information hiding scheme, as follows (see Figure~\ref{fig:DWT}). Let:

\begin{itemize}
  \item $(K,N) \in [0;1]\times \mathds{N}$ be an embedding key,
  \item $X \in \mathbb{B}^\mathsf{N}$ be the $\mathsf{N}$ least significant coefficients (LSCs) of a given cover media $C$,
  \item $(S^n)_{n \in \mathds{N}} \in \llbracket 1, \mathsf{N} \rrbracket^{\mathds{N}}$ be a strategy, which depends on the message to hide $M \in [0;1]$ and $K$, 
  \item $f_0 : \mathbb{B}^\mathsf{N} \rightarrow \mathbb{B}^\mathsf{N}$ be the vectorial logical negation.
\end{itemize}

\begin{figure}[htb]
\begin{minipage}[b]{.45\linewidth}
  \centering
 \centerline{\epsfig{figure=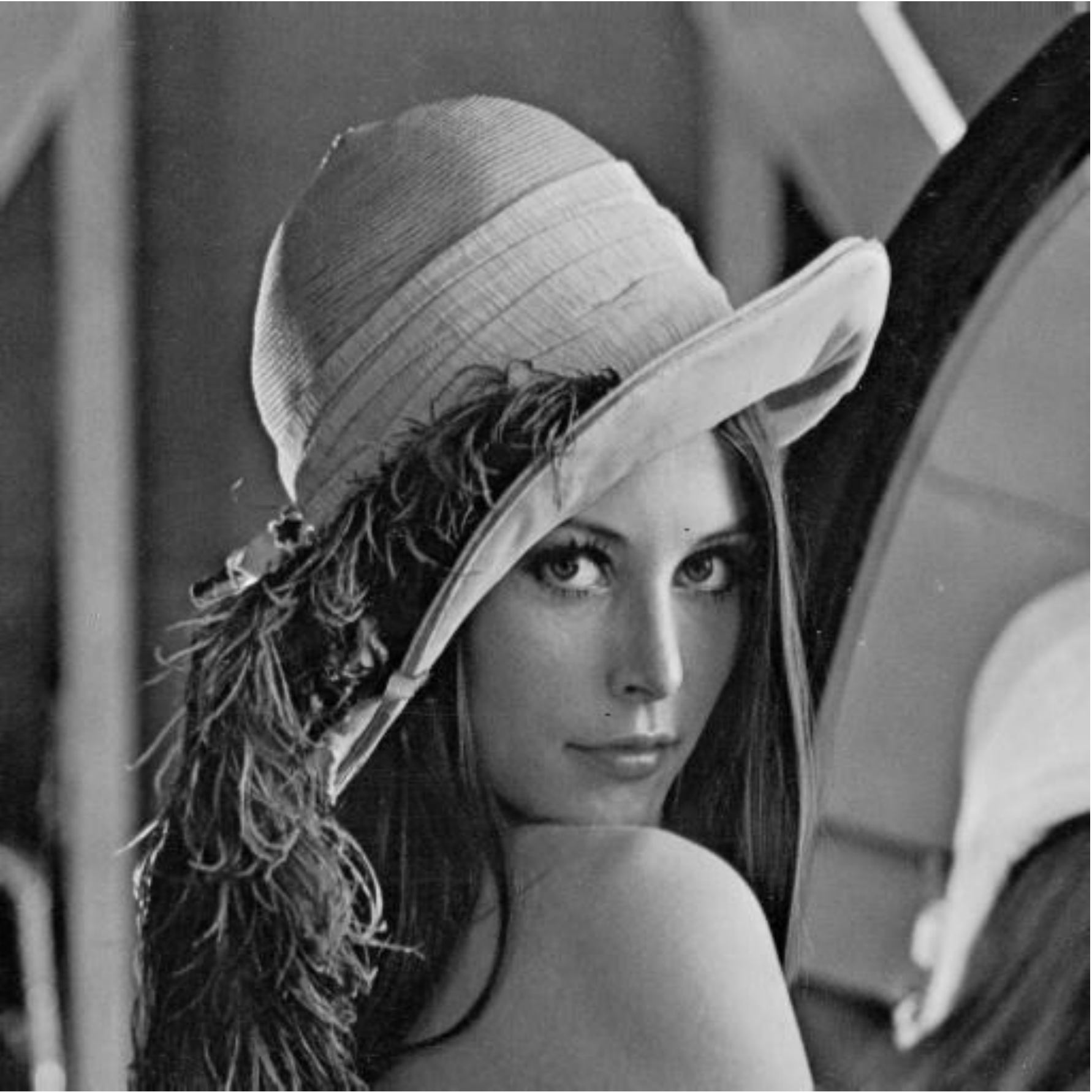,width=3.5cm}}
  \centerline{(a) Original Lena.}
\end{minipage}
\hfill
\begin{minipage}[b]{0.45\linewidth}
  \centering
 \centerline{\epsfig{figure=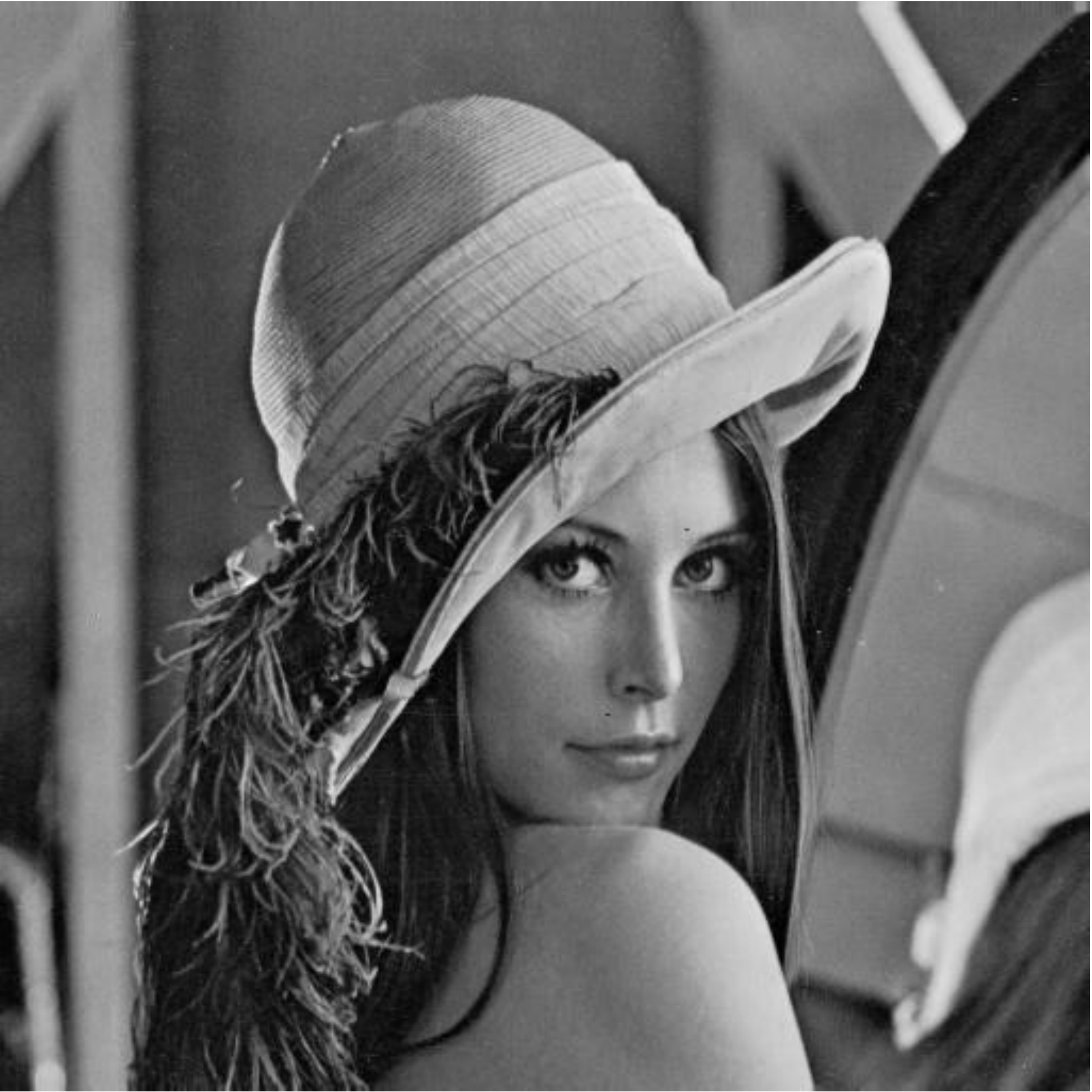,width=3.5cm}}
  \centerline{(b) Watermarked Lena.}
\end{minipage}
\caption{Data hiding with chaotic iterations}
\label{fig:DWT}
\end{figure}


So the watermarked media is $C$ whose LSCs are replaced by $Y_K=X^{N}$, where:

\begin{equation*}
\left\{
 \begin{array}{l}
X^0 = X\\
\forall n < N, X^{n+1} = G_{f_0}\left(X^n\right).\\
\end{array} \right.
\end{equation*}

In the following section, two ways to generate $(S^n)_{n \in \mathds{N}}$ are given, namely
Chaotic Iterations with Independent Strategy~(CIIS) and Chaotic Iterations with Dependent
Strategy~(CIDS). In CIIS, the strategy is independent from the cover media $X$,
whereas in CIDS the strategy will be dependent on $X$. Their stego-security are
studied in Section~\ref{sec:chaotic iterations-security-level} and their
topological-security in Section~\ref{sec:chaos-security-evaluation}.

\subsubsection{Examples of strategies}\label{sec:ciis-cids-example}
\paragraph{CIIS strategy}

Let us first introduce the Piecewise Linear Chaotic Map~(PLCM, see~\cite{Shujun1}), defined by:

\begin{definition}[PLCM]\label{Def:PLCM}
\begin{equation*}
F(x,p)=\left\{
 \begin{array}{ccc}
x/p & \text{if} & x \in [0;p] \\
(x-p)/(\frac{1}{2} - p) & \text{if} & x \in \left[ p; \frac{1}{2} \right]
\\
F(1-x,p) & \text{else.} & \\
\end{array} \right.
\end{equation*}
\end{definition}

\noindent where $p \in \left] 0; \frac{1}{2} \right[$ is a ``control parameter''.


Then, we can define the general term of the strategy $(S^n)_n$ in CIIS setup by
the following expression: $S^n = \left \lfloor \mathsf{N} \times K^n \right \rfloor +
1$, where:

\begin{equation*}\label{eq:PLCM-for-strategy}
\left\{
 \begin{array}{l}
p \in \left[ 0 ; \frac{1}{2} \right] \\
K^0 = M \otimes K\\
K^{n+1} = F(K^n,p), \forall n \leq N_0\\ \end{array} \right.
\end{equation*}

\noindent in which $\otimes$ denotes the bitwise exclusive or (XOR) between two floating part numbers (\emph{i.e.}, between their binary digits representation). Lastly, to be certain to enter into the chaotic regime of PLCM~\cite{Shujun1}, the strategy can be preferably defined by: $S^n = \left \lfloor \mathsf{N} \times K^{n+D} \right \rfloor + 1$, where $D \in \mathds{N}$. 

\paragraph{CIDS strategy}\label{sec:cids-example}
The same notations as above are used.
We define CIDS strategy as follows: $\forall k \leqslant N$,  
\begin{itemize}
\item if $k \leqslant \mathsf{N}$ and $X^k = 1$, then $S^k=k$,
\item else $S^k=1$.
\end{itemize}
In this situation, if $N \geqslant \mathsf{N}$, then only two watermarked contents are possible with the scheme proposed in Section~\ref{sec:data-hiding-algo-chaotic iterations}, namely: $Y_K=(0,0,\cdots,0)$ and $Y_K=(1,0,\cdots,0)$.


\section{Evaluation of the stego-security}\label{sec:chaotic iterations-security-level}

\subsection{Definition of
stego-security}\label{sec:stego-security-definition}

Stego-security, defined in the Simmons' prisoner problem~\cite{Simmons83}, is 
the highest security class in WOA setup~\cite{Cayre2008}.

Let $\mathds{K}$ be the set of embedding keys, $p(X)$ the probabilistic
model of $N_0$ initial host contents, and $p(Y|K_1)$ the probabilistic
model of $N_0$ watermarked contents. We suppose that each host content has been 
watermarked with the same key $K_1$ and the same embedding function $e$.

\begin{definition}
\label{Def:Stego-security}
The embedding function $e$ is stego-secure if and only if:
$$\mathbf{\forall K_1 \in \mathds{K}, p(Y|K_1)=p(X)}$$
\end{definition}



\subsection{Evaluation of the stego-security}
\label{sec:stego-security-proof}

Let us now study the stego-security of the scheme. We will prove that,

\begin{proposition}
CIIS are stego-secure.
\end{proposition}

\begin{proof}
Let us suppose that $X \sim
\mathbf{U}\left(\mathbb{B}^N\right)$ in a CIIS setup. 
We will prove by a mathematical induction that $\forall n \in \mathds{N}, X^n \sim
\mathbf{U}\left(\mathbb{B}^N\right)$. The base case is immediate, as $X^0 = X \sim
\mathbf{U}\left(\mathbb{B}^N\right)$. Let us now suppose that the statement
 $X^n \sim
\mathbf{U}\left(\mathbb{B}^N\right)$ holds for some $n$. 
Let $e \in \mathbb{B}^N$ and $\mathbf{B}_k=(0,\cdots,0,1,0,\cdots,0) \in \mathbb{B}^N$ (the digit $1$ is in position $k$).
So $P\left(X^{n+1}=e\right)=\sum_{k=1}^N
P\left(X^n=e+\mathbf{B}_k,S^n=k\right).$
These two events are independent in CIIS setup, thus:
$P\left(X^{n+1}=e\right)=\sum_{k=1}^N
P\left(X^n=e+\mathbf{B}_k\right) \times P\left(S^n=k\right)$.
According to the inductive hypothesis:
$P\left(X^{n+1}=e\right)=\frac{1}{2^N} \sum_{k=1}^N
 P\left(S^n=k\right)$.
The set of events $\left \{ S^n=k \right \}$ for $k \in \llbracket 1;N
\rrbracket$ is a partition of the universe of possible, so
$\sum_{k=1}^N P\left(S^n=k\right)=1$.

Finally, 
$P\left(X^{n+1}=e\right)=\frac{1}{2^N}$, which leads to $X^{n+1} \sim
\mathbf{U}\left(\mathbb{B}^N\right)$.
This result is true $\forall n \in \mathds{N}$, we thus have proven that, $$\forall K \in [0;1], Y_K=X^{N_0} \sim
\mathbf{U}\left(\mathbb{B}^N\right) \text{ when } X \sim
\mathbf{U}\left(\mathbb{B}^N\right)$$

So CIIS defined in
Section~\ref{sec:data-hiding-algo-chaotic iterations} are stego-secure.
\end{proof}

We will now prove that,

\begin{proposition}
CIDS are not stego-secure.
\end{proposition}

\begin{proof}
Due to the definition of CIDS, we have $P(Y_K=(1,1,\cdots,1))=0$. So there is
no uniform repartition for the stego-contents $Y_K$.
\end{proof}

\section{Evaluation of the
topological-security}\label{sec:chaos-security-evaluation}

\subsection{Definition}\label{sec:chaos-security-definition}


To check whether an information hiding scheme $S$ is topologically-secure or
not, $S$ must be written as an iterate process $x^{n+1}=f(x^n)$ on a metric space
$(\mathcal{X},d)$. This formulation is always possible, as it is proven
in~\cite{ih10}. So,

\begin{definition}
\label{Def:chaos-security-definition}
An information hiding scheme $S$ is said to be topologically-secure on
$(\mathcal{X},d)$ if its iterative process has a chaotic
behavior according to Devaney.

\end{definition}

It can be established that,

\begin{proposition}
CIIS and CIDS are topologically-secure.
\end{proposition}

\begin{proof}
It has been proven in~\cite{bg10:ij} that chaotic iterations have a chaotic behavior,
as defined by Devaney.
\end{proof}

In the two following sections, we will study the qualitative and quantitative
properties of topological-security for chaotic iterations. These properties can
measure the disorder generated by our scheme, giving by doing so some important informations about the
unpredictability level of such a process.

\subsection{Quantitative property of
chaotic iterations}\label{sec:chaos-security-quantitative}

\begin{definition}[Expansivity]
A function $f$ is said to be \emph{expansive} if $
\exists \varepsilon >0,\forall x\neq y,\exists n\in \mathds{N}%
,d(f^{n}(x),f^{n}(y))\geqslant \varepsilon .$
\end{definition}



\begin{proposition}
 $G_{f_{0}}$ is an expansive chaotic dynamical system on $\mathcal{X}$ with a constant of expansivity is equal to 1.
\end{proposition}
\begin{proof}
If $(S,E)\neq (\check{S};\check{E})$, then either $E\neq \check{E}$, so at least
one cell is not in the same state in $E$ and $\check{E}$. Consequently the distance between $(S,E)$ and $(%
\check{S};\check{E})$ is greater or equal to 1. Or $E=\check{E}$. So the
strategies $S$ and $\check{S}$ are not equal. Let $n_{0}$ be the first index in which the terms $S$ and $\check{S}$
differ. Then $
\forall
k<n_{0},\tilde{G}_{f_{0}}^{k}(S,E)=\tilde{G}_{f_{0}}^{k}(\check{S},\check{E})$,
and $\tilde{G}_{f_{0}}^{n_{0}}(S,E)\neq \tilde{G}_{f_{0}}^{n_{0}}(\check{S},\check{E})$. As $E=\check{E},$ the cell which has changed in $E$ at the $n_{0}$-th
iterate is not the same as the cell which has changed in $\check{E}$, so
the distance between $\tilde{G}_{f_{0}}^{n_{0}}(S,E)$ and $\tilde{G}_{f_{0}}^{n_{0}}(\check{S%
},\check{E})$ is greater or equal to 2.
\end{proof}

\subsection{Qualitative property of
chaotic iterations}\label{sec:chaos-security-qualitative}

\begin{definition}[Topological mixing]
A discrete dynamical system is said to be topologically mixing
if and only if, for any couple of disjoint open set $U, V \neq \varnothing$,
$n_0 \in \mathds{N}$ can be found so that $\forall n \geqslant n_0, f^n(U) \cap
V \neq \varnothing$.
\end{definition}
\begin{proposition}
$\tilde{G}_{f_0}$ is topologically mixing on $(\mathcal{X}', d')$.
\end{proposition}

This result is an immediate consequence of the lemma below.

\begin{lemma}
For any open ball $B$ of $\mathcal{X}'$, an index $n$ can be found such that $\tilde{G}_{f_0}^n(B) = \mathcal{X}'$.
\end{lemma}

\begin{proof}
Let $B=B((E,S),\varepsilon)$ be an open ball, which the radius can be considered as strictly less than 1.
All the elements of $B$ have the same state $E$ and are such that an integer $k \left(=-\log_{10}(\varepsilon)\right)$ satisfies:
\begin{itemize}
\item all the strategies of $B$ have the same $k$ first terms,
\item after the index $k$, all values are possible.
\end{itemize}

Then, after $k$ iterations, the new state of the system is $\tilde{G}_{f_0}^k(E,S)_1$ and all the strategies are possible (all the points $(\tilde{G}_{f_0}^k(E,S)_1,\textrm{\^{S}})$, with any $\textrm{\^{S}} \in \mathbb{S}$, are reachable from $B$).

We will prove that all points of $\mathcal{X}'$ are reachable from $B$. Let
$(E',S') \in \mathcal{X}'$ and $s_i$ be the list of the different cells between
$\tilde{G}_{f_0}^k(E,S)_1$ and $E'$. We denote by $|s|$ the size of the sequence
$s_i$. So the point $(\check{E},\check{S})$ of $B$ defined by: $\check{E} = E$,
$\check{S}^i = S^i, \forall i \leqslant k$, $\check{S}^{k+i} = s_i, \forall i
\leqslant |s|$, and $\forall i \in \mathds{N}, S^{k + |s| + i} = S'^i$
is such that $\tilde{G}_{f_0}^{k+|s|}(\check{E},\check{S}) = (E',S')$. This concludes the proofs of the lemma and of the proposition.
\end{proof}
\section{Comparison between
spread-spectrum and chaotic
iterations}\label{sec:comparison-application-context}

The consequences of topological mixing for data hiding are multiple. Firstly, 
security can be largely improved by considering
the number of iterations as a secret key. An attacker
will reach all of the possible media when iterating without this key.
Additionally, he cannot benefit from a KOA setup, by studying media in the
neighborhood of the original cover. Moreover, as in a topological mixing
situation, it is possible that any hidden message (the initial condition), is
sent to the same fixed watermarked content (with different numbers of
iterations), the interest to be in a KMA setup is drastically reduced. Lastly, as
all of the watermarked contents are possible for a given hidden message, depending
on the number of iterations, CMA attacks will fail.

The property of expansivity reinforces drastically the sensitivity in the aims of
reducing the benefits that Eve can obtain from an attack in KMA or KOA setup. For
example, it is impossible to have an estimation of the watermark by moving the
message (or the cover) as a cursor in situation of expansivity: this cursor will
be too much sensitive and the changes will be too important to be useful. On the
contrary, a very large constant of expansivity $\varepsilon$ is unsuitable: the
cover media will be strongly altered whereas the watermark would be undetectable.


Finally, spread-spectrum is relevant when
a discrete and secure data hiding technique is required in WOA setup. However,
this technique should not be used in KOA and KMA setup, due to its lack of
expansivity.
schemes, which are expansive.

\section{Conclusion and future work}\label{sec:conclusion}


In this paper, the links between stego-security and topological-security has
been deepened. The information hiding scheme presented in~\cite{guyeux10ter}, which is
based on chaotic iterations, has been recalled and its level of security has been
studied. It has been proven that this algorithm is twice stego and
topologically-secure. This was already the case for spread-spectrum techniques,
as it has been established in~\cite{ih10}. Moreover, as for spread-spectrum, chaotic iterations
possess the qualitative property of topological mixing, which are useful to
withstand attacks. However, unlike spread-spectrum, chaotic iterations are
expansive, so this scheme is better than spread-spectrum in KOA and KMA setups.
Incidentally, this result shows that the new framework for security tends to
improve the ability to compare data hiding scheme. In future work, we will give a
better understanding of the links between these two security frameworks.
Additionally, the comparison between spread-spectrum and chaotic iterations
outlined in this paper will be extended. The security of other existing schemes
will be studied in the framework of topological-security. Last, but not least,
the way to understand these new tools in terms of data hiding aims will be enhanced: this
study is required to make topological-security framework truly useful in
practice.

\newpage
\bibliographystyle{plain}
\bibliography{jabref}

\end{document}